\documentclass[letterpaper, 10 pt, conference]{ieeeconf}  

\IEEEoverridecommandlockouts                              

\usepackage{graphics} 
\usepackage{amsmath} 
\usepackage{amssymb}  
\usepackage[subnum]{cases}
\usepackage{lipsum}
\usepackage{subcaption}
\usepackage{float}
\usepackage{soul} 
\usepackage{bm}
\usepackage{framed}
\usepackage{graphicx}
\usepackage{hyperref}
\hypersetup{
    colorlinks=true,
    linkcolor=blue,
    filecolor=magenta,      
    urlcolor=cyan,
}

\usepackage{cite}

\usepackage{scrextend}

\usepackage[ruled,linesnumbered]{algorithm2e}
\usepackage{algpseudocode}
\usepackage{listings}
\usepackage{lipsum}
\usepackage{adjustbox}

\newtheorem{theorem}{Theorem}

\newtheorem{lemma}{Lemma}

\newtheorem{example}{Example}
\newtheorem{remark}{Remark}
\newtheorem{assumption}{Assumption}

\usepackage[markup=bfit, deletedmarkup=sout, authormarkup=brackets]{changes}
\definechangesauthor[name={GD}, color=cyan]{GD}
\definechangesauthor[name={TODO}, color=red]{TODO}

\DeclareMathOperator{\tr}{\textbf{tr}}
\DeclareMathOperator{\vect}{\textbf{vec}}

\title{\LARGE \bf
Towards Scalable Koopman Operator Learning: Convergence Rates and A Distributed Learning Algorithm
}

\author{Zhiyuan Liu, Guohui Ding, Lijun Chen\thanks{Z. Liu, G. Ding and L. Chen are with the Department of Computer Science, University of Colorado, Boulder, CO 80309, USA (emails: \{zhiyuan.liu, duohui.ding,  lijun.chen\}@colorado.edu).} and  Enoch Yeung \thanks{E. Yeung is with the Department of Mechanical Engineering, the Center for Control, Dynamical Systems, and the Center for Biological Engineering, University of California,  Santa Barbara, CA 93106, USA (email: eyeung@ucsb.edu).}
}

\begin{document}

\maketitle
\thispagestyle{empty}
\pagestyle{empty}

\begin{abstract}
We propose an alternating optimization algorithm to the nonconvex Koopman operator learning problem for nonlinear dynamic systems. We show that the proposed algorithm will converge to a critical point with rate $O(1/T)$ and $O(\frac{1}{\log T})$ for the constant and diminishing learning rates, respectively, under some mild conditions. To cope with the high dimensional nonlinear dynamical systems, we present the first-ever distributed Koopman operator learning algorithm. We show that the distributed Koopman operator learning has the same convergence properties as  the centralized Koopman operator learning, in the absence of optimal tracker, so long as the basis functions satisfy a set of state-based decomposition conditions. Numerical experiments are provided to complement our theoretical results.
\end{abstract}

\section{Introduction}\label{sect:intro}
There is an increasing interest in recent years in transferring the operator theoretic techniques such as the Koopman operator \cite{koopman1931hamiltonian,mezic2005spectral} to  the analysis of dynamical systems. Such operator based methods differ from  the classical approaches, in that they define the evolution of observable functions in a function space rather than using state vectors in the state space. The power of these operator theoretic methods is that it provides linear representations of nonlinear time-invariant systems, albeit in higher dimensional spaces that are countable or uncountable.  Various numerical approaches, such as dynamic mode decomposition(DMD), Hankel-DMD, extended dynamic mode decomposition (E-DMD) and structured dynamic mode decomposition (S-DMD),  have been proposed for discovering the Koopman operator of a nonlinear system, using a series of dictionary functions with spanning or universal function approximation properties \cite{li2017extended,kutz2016multiresolution,mezic2005spectral, rowley2009spectral, sinha2019computation}.   
Researchers have recently shown that it is possible to integrate machine-learned representations with dynamic mode decomposition algorithms, using variational autoencoders to achieve phase-dependent representations of spectra \cite{lusch2018deep} or delay embeddings \cite{takeishi2017learning}, shallow neural networks \cite{li2017extended}, linearly recurrent neural networks for balancing expressiveness and overfitting \cite{otto2019linearly}, and deep RELU feedforward networks for predictive modeling in biological and transmission systems \cite{yeung2019learning}.  
E-DMD \cite{li2017extended} and Deep-DMD \cite{yeung2019learning}  have been utilized in various domains, including nonlinear system identification \cite{johnson2018class,liu2018decomposition,mehta2005stochastic,yeung2019learning}, image processing \cite{kutz2016multiresolution,kutz2015multi} and robotic control \cite{abraham2017model,berger2015estimation}.

Generally speaking, the learning especially the training phase of the Koopman operator tries to minimize the empirical loss based on the training set, e.g., the data sampled from the real trajectory of dynamic system. Compared to the traditional machine learning problem which learns the unknown mapping from input to output, the Koopman learning has two tasks: 1) Learning the function space that lifts state space to a high even infinite dimensional space, and 2) learning a linear mapping within that function space. These two tasks are highly intertwined, e.g., inappropriate function space learned will lead to poor learning performance even the linear mapping is perfect. However, to the best of our knowledge, the method of Koopman training has not gotten enough attention up to now. Another challenge is that, when parameterized function approximation such as neural network is used,  the learning problem is nonconvex. For instance, even for a single layer neural network, it is NP-complete to find the global optimal \cite{blum1989training}. However, recent work \cite{soltanolkotabi2018theoretical,boob2017theoretical,choromanska2015loss} shows that for over-parameterized (wide) shallow neural networks, local optima provide satisfactory performance.  Specifically, they show that every local optimum is global optimum if the hidden layer is non-singular; every local minimum of the simplified objective is close to the global minimum. In this paper, we contribute a proof of convergence for Koopman learning algorithms utilizing shallow neural networks, and derive conditions for first-order optimality, the properties of the so-called dictionary functions used in deep and E-DMD that guarantee convergence. We propose alternating optimization algorithm with an optimal tracker for training the Koopman operator. By proving the objective function's  smoothness property, we show that our algorithm admits $O(1/T)$ convergence rate for chosen constant learning rate and $O(1/\log T)$ for diminishing learning rate. We illustrate convergence of the alternating optimization algorithm for single-node training (non-distributed) on two  nonlinear systems with oscillatory dynamics.

A second major contribution of this paper is the development of a distributed Koopman operator learning algorithm.  Most Koopman operator learning algorithms operate under the assumption of full-state measurements.  However, in engineered and natural systems represented by data, full-state measurements are often not available, or are too expensive to collect.  For example, power distribution networks consisting of hundreds of thousands of nodes exhibit real-time dynamics on systems that are poorly modeled, calibrated, or dated.  Biological networks operate on thousands of genes to generate transcriptomic reponse profiles as a function of time; full-state measurement via deep sequencing is prohibitively expensive.  In many instances, it is much more feasible to collect measurements from select locations, via strategic placement of observers \cite{hasnain2019optimal}, which gives rise to a different form of data --- time-series data that is spatially distributed or fragmented across the whole network.  We address the challenge of training distributed representations of Koopman operators and develop a {\it distributed} Koopman learning algorithm, proving its asymptotic convergence, and illustrating predictive accuracy and convergence on several simulated examples. 

The  rest  of  the  paper  is  organized  as  follows.  Section  II introduces the Koopman operator learning problem. Section  III describes our alternating optimization algorithm for the Koopman learning and proves the convergence. Section   IV extends the algorithm to a distributed setting and shows its convergence. Section V evaluates the performance of the  two algorithms on two nonlinear systems. Section  VI concludes the paper.

\section{Koopman operator Learning Problem} \label{sect:koopman}
We consider a discrete time open-loop nonlinear dynamic system of the following form: 
	\begin{align}
		x_{n+1} &= f(x_{n}),  \label{equ:1}\\
		y_n &= h(x_n),
	\end{align}
where $f: \mathbb{R}^d \rightarrow \mathbb{R}^{d} $ and $h: \mathbb{R}^{d} \rightarrow \mathbb{R}^{p}$ are continuously differentiable. The function $f$ is the state-space model and the function $h$ maps current state $x_n \in \mathbb{R}^{d}$ to a vector of observables or output $y_n \in \mathbb{R}^{p}$.  The Koopman operator $\mathcal{K}$ of system 
\eqref{equ:1}, if it exists, is a linear operator that acts on observable functions $\psi(x_k)$ and forward propagates them in time. To be more specific, the Koopman operator for this system must satisfy the equations: 
\begin{align}
    \psi(x_{n+1}) &= \mathcal{K}(\psi(x_n)), \\
		y_n &= \mathcal{H}(\psi(x_n)),
\end{align}
where $\psi(x_n) = [\psi_1(x_n),\cdots,\psi_m(x_n)]^{\top}: \mathbb{R}^{d} \rightarrow \mathbb{R}^{m} (m \leq \infty)$ is a basis function that defines the lifted space of observables and $\mathcal{K} \in \mathbb{R}^{m\times m}$ is a constant matrix. Based on the Koopman operator theory, $\psi$ is the basis function of observables under which $\psi(x_n)$ is $\mathcal{K}$-invariant for all $n$. This implies that the Koopman operator comprehensively captures the flow of the observable trajectory $(x_1,x_2,\cdots)$. 

Based on the data-driven method \cite{yeung2019learning} \cite{williams2015data}, a general model for approximating Koopman operator given the data trajectory $(x_i,x_{i+1}), i \in \{1,\cdots,N\}$ can be formulated as follows:
	\begin{align}\label{eq:lm}
		\min_{\psi,\mathcal{K}}~ \mathcal{D}(\psi,\mathcal{K}):=\frac{1}{2N}\sum_{i=1}^{N}\|\psi(x_{i+1}) - \mathcal{K}\psi(x_i)\|_2^2.
	\end{align}
The above model aims to minimize the empirical loss from the learning perspective. One can slightly change the objective function by adding certain regularized term, e.g., $\|\mathcal{K}\|_1$ for sparse operator or $\|\mathcal{K}\|_2$ for avoiding large training lost, to make the tradeoff between the training and generalization errors.

While there has been a surge of interest in using neural networks to perform Koopman learning, little is known regarding the convergence and numerical stability of the training process. This motivates us to investigate the property of Koopman learning during its training phase. There are two challenges in solving  the optimization problem \eqref{eq:lm} in practice. First, the basis function $\psi$ is unknown. This makes it difficult to ascertain what functions and how many functions to include, let alone the minimal number of functions, to ensure $\mathcal{K}$-invariant. Recently, EDMD \cite{williams2015data} uses an expansive set of orthonomal polynomial basis functions, but this approach does not scale well and suffers from overfitting with an increasing number of basis functions. Deep-DMD \cite{yeung2019learning} adopts the neural networks to approximate the basis function based on universal approximation theorem, but it lacks the theoretical guarantee, e.g., the stability and convergence. Second, the objective function is nonconvex. Therefore it is unrealistic to expect an algorithm to converge to global minima. 

Here we focus on the basis function based on parametric method. Specifically, we redefine $\psi(x_n)= \psi(\mathcal{W}x_n)$. The term $\mathcal{W}x_n$ means matrix product. A typical example is a fully connected one-layer neural network since for wide shallow neural network, local optima provide satisfactory under some mild conditions \cite{soltanolkotabi2018theoretical,boob2017theoretical,choromanska2015loss}, where $\mathcal{W}$ is the layer parameter and $\psi$ is activation function. With the parametric basis method, problem \eqref{eq:lm} becomes
\begin{equation}
	\begin{aligned}
		\min_{\mathcal{W},\mathcal{K}}~ \mathcal{F}(\mathcal{W},\mathcal{K}):=\frac{1}{2N}\sum_{i=1}^{N}\|\psi(\mathcal{W}x_{i+1}) - \mathcal{K}\psi(\mathcal{W}x_i)\|_2^2.
	\end{aligned}
\end{equation}
Although this problem is nonconvex, there are some interesting structures. For example, if we fix the parameter $\mathcal{W}$ of the basis function, optimizing $\mathcal{K}$ is a quadratic problem that finds the linear mapping from $\mathbb{R}^{m}$ to $\mathbb{R}^{m}$. On the other hand, with fixed $\mathcal{K}$, optimizing $\mathcal{W}$ is to adjust the parameter $\mathcal{W}$ to find the function space that satisfies the linear transformation mapping (this is still highly nonconvex but will reduce the complexity a lot). We thus consider the algorithm that alternatively optimize over $\mathcal{W}$ and $\mathcal{K}$. 

\section{Alternating Optimization Algorithm}\label{sect:alternating}

In this section, we first state our alternating algorithm and then investigate its convergence properties. Let $\mathcal{F}_{t} = \mathcal{F}(\mathcal{W}^{t},\mathcal{K}^{t})$ and  denote by $\|\cdot\|_{F}$ the Frobenius norm. The detail of the algorithm is shown in Algorithm \ref{algo:0}. Here $\mathcal{E}$ measures how far the gradient is from that at the critical point and $\mathcal{K}^{*},\mathcal{W}^{*}$ track the best parameters so far. 
We make the following assumptions.

\begin{algorithm}
	\caption{Alternating Operator Koopman Learning With Tracking}\label{algo:0}
	\textbf{Initialization:}
	randomly initialize $\mathcal{W}^{0}$ and $\mathcal{K}^{0}$, $\mathcal{E}^{0} = \|\nabla_{\mathcal{K}}\mathcal{F}_{0}\|_{F} +\|\nabla_{\mathcal{W}}\mathcal{F}_{0}\|_{F} $,$\mathcal{W}^{*} = \mathcal{W}^{0}$,$\mathcal{K^{*}} = \mathcal{K}^{0}$. \\
	\While{Not Converge}{ 
		$\mathcal{K}^{t+1} ~= \mathcal{K}^{t} - \eta_{\mathcal{K}} \nabla_{\mathcal{K}} \mathcal{F}(\mathcal{W}^t, \mathcal{K}^t),$ \\
		$\mathcal{W}^{t+1} = \mathcal{W}^{t} - \eta_{\mathcal{W}} \nabla_{\mathcal{W}} \mathcal{F}(\mathcal{W}^t, \mathcal{K}^{k+1}).$\\
		$\mathcal{E}^{t+1} ~= \|\nabla_{\mathcal{K}}\mathcal{F}_{t+1}\|_{F}+\|\nabla_{\mathcal{W}}\mathcal{F}_{t+1}\|_{F}.$ \\
		\If{$\mathcal{E}^{t+1} \leq \mathcal{E}^{t} $
		}{
			$\mathcal{K}^{*} = \mathcal{K}^{t+1};$
			$\mathcal{W}^{*} = \mathcal{W}^{t+1}$  	
		}
	}	
\end{algorithm}
\begin{assumption}\label{Assumption:1}
	The function $\psi(\cdot)$ is bounded and has a bounded gradient and Hessian.
\end{assumption}
\begin{assumption} \label{Assumptipon:2}
	The parameters $\mathcal{K}$ and $\mathcal{W}$ are bounded, i.e., there exist two constant $U_{\mathcal{K}}$ and $U_{\mathcal{W}}$ such that $\|\mathcal{K}\|_{F} \leq U_{\mathcal{K}}$ and $\|\mathcal{W}\|_{F} \leq U_{\mathcal{W}}$.
\end{assumption}

	Assumption \ref{Assumption:1} looks strong.  However, it holds for several popular activation functions such as logistic function ($\frac{1}{1+e^{-x}}$), hyperbolic tangent ($\tanh(x)$), and inverse hyperbolic tangent ($\arctan(x))$. By Assumptions \ref{Assumption:1} and \ref{Assumptipon:2}, one can verify that the objective function $\mathcal{F}$ is bounded, i.e., there exists a constant $R$ such that $\mathcal{F} \leq R$. 
	We can show that $\mathcal{F}$ has Lipschitz-continuous gradient with respect to the parameter $\mathcal{W}$ of basis functions. 
	\begin{lemma} \label{lemma:1}
	Under Assumptions \ref{Assumption:1}  and \ref{Assumptipon:2} and given the data trajectory $\{(x_i,x_{i+1})\}_{i=1}^{N}$,  we have 
	\begin{align}
		\|\nabla_{\mathcal{W}} \mathcal{F}(\mathcal{W}^1,\mathcal{K}) - \nabla_{\mathcal{W}} \mathcal{F}(\mathcal{W}^2,\mathcal{K}) \|_{F} \leq L_{\mathcal{W}}\|\mathcal{W}^{1} - \mathcal{W}^2\|_{F}  \nonumber
	\end{align}
	with 
	\begin{align}
		L_{\mathcal{W}} =  \sqrt{2d}U_{\mathcal{K}} L_{\Psi} \frac{\sum_{i=1}^{N}\|x_i\|_2 \Delta_i}{N}, \nonumber
	\end{align}
	 where $\Delta_i =  \sqrt{(1+dU_\mathcal{K}^2)\|x_i\|^2_2 + \|x_{i+1}\|_2^2}$ and $L_{\Psi}$ is the Lipschitz constant for the function $\Psi(x_1,x_2):= \psi(x_1)\psi'(x_2)$.
\end{lemma}

\begin{proof}
	First, denote by $\mathcal{K}[:,i]$ the $i$-th column of matrix $\mathcal{K}$, $\mathcal{W}_j$ the $j$-th row of matrix $\mathcal{W}$,  and $x_i[k]$ the $k$-th dimension of $x_i$.  We can compute the element $[j,k]$ of $\nabla_\mathcal{W} \mathcal{F}(\mathcal{W},\mathcal{K})$ as:
	\begin{equation*}
	\begin{aligned}
	&\nabla_\mathcal{W} \mathcal{F}(\mathcal{W},\mathcal{K})[j,k] \\
	=& \frac{1}{N}\sum_{i=1}^{N} -(\psi(\mathcal{W}x_{i+1})-\mathcal{K} \psi(\mathcal{W}x_i))^T\mathcal{K}[:,j]\psi'(\mathcal{W}x_i)x_i[k],
	\end{aligned}
	\end{equation*}
	and $\nabla_\mathcal{W} \mathcal{F}(\mathcal{W},\mathcal{K})$ as: 
	\begin{equation*}
	\begin{aligned}
	&\nabla_\mathcal{W} \mathcal{F}(\mathcal{W},\mathcal{K}) \\
	=& -\frac{1}{N}\sum_{i=1}^{N}\mathcal{K}^T \underbrace{(\psi(\mathcal{W}x_{i+1}) - \mathcal{K} \psi(\mathcal{W}x_i)) \odot \psi'(\mathcal{W}x_i)}_{\alpha_i^{\mathcal{W}}}x_i^T\\
	=& -\frac{1}{N}\sum_{i=1}^{N}\mathcal{K}^T \alpha_i^{\mathcal{W}} x_i^T,
	\end{aligned}
	\end{equation*}
	where $\odot$ denotes the element-wise production. We can then write the gradient difference with respect to $\mathcal{W}^1$ and $\mathcal{W}^2$ as
	\begin{equation*}
	\begin{aligned}
	&\| \nabla_\mathcal{W} \mathcal{F}(\mathcal{W}^1, \mathcal{K}) - \nabla_\mathcal{W} \mathcal{F}(\mathcal{W}^2, \mathcal{K}) \|_F\\
	=& \frac{1}{N} \sum_{i=1}^{N}\|\mathcal{K}^T(\alpha_i^{\mathcal{W}^1}-\alpha_i^{\mathcal{W}^2})x_i^T\|_F\\
	\leq & \frac{1}{N} \sum_{i=1}^{N} \|\mathcal{K}\|_F\|(\alpha_i^{\mathcal{W}^1}-\alpha_i^{\mathcal{W}^2})x_i^T\|_F \\
	\leq & \frac{1}{N} \sum_{i=1}^{N} \|\mathcal{K}\|_F\|\alpha_i^{\mathcal{W}^1}-\alpha_i^{\mathcal{W}^2} \|_2 \|x_i\|_2.
	\end{aligned}
	\end{equation*}
	So if we can show that $\alpha_i^{\mathcal{W}}$ is Lipschitz-continuous, the proof is done. We have
	\begin{equation*}
	\begin{aligned}
	&\alpha_i^{\mathcal{W}^1}[j] - \alpha_i^{\mathcal{W}^2}[j] \\
	=& \underbrace{\left(\psi(\mathcal{W}_j^1x_{i+1})\psi'(\mathcal{W}_j^1x_i)-\psi(\mathcal{W}_j^2x_{i+1})\psi'(\mathcal{W}_j^2x_i)\right)}_{\beta_j^i} \\
	& ~~- \Big(\underbrace{\mathcal{K}_j \psi(\mathcal{W}^1x_i)\psi'(\mathcal{W}_j^1x_i)-\mathcal{K}_j \psi(\mathcal{W}^2x_i)\psi'(\mathcal{W}_j^2x_i)}_{\gamma_j^i}\Big).
	\end{aligned}
	\end{equation*} 
	Consider function $\Psi(x_1,x_2) = \psi(x_1)\psi'(x_2)$ in $x_1,x_2\in \mathbb{R}$ and its gradient $\nabla \Psi(x_1,x_2) = \left[\begin{matrix}
	\psi'(x_1)\psi'(x_2)\\
	\psi(x_1)\psi''(x_2)
	\end{matrix} \right].$ By Assumption \ref{Assumption:1}, $\| \nabla \Psi(\cdot) \|_2$ is bounded by some constant, denoted by $L_{\Psi}$.
	Let $\beta^i = [\beta_1^i, \cdots, \beta_n^i]^T$, we can bound $\beta^i$ as follows:
	\begin{equation*}
	\begin{aligned}
	\|\beta^i\|_2^2 &=\! \|\psi(\mathcal{W}^1x_{i+1}) \odot \psi'(\mathcal{W}^1x_i) \!-\! \psi(\mathcal{W}^2x_i) \odot \psi'(\mathcal{W}^2x_i) \|_2^2 \\
	& \leq 2L_{\Psi}^2\Big(\|\mathcal{W}^1x_{i+1} - \mathcal{W}^2x_{i+1}\|_2^2 + \|\mathcal{W}^1x_i - \mathcal{W}^2x_i\|_2^2 \Big) \\
	& \leq 2L_{\Psi}^2 (\|x_i\|_2^2 + \|x_{i+1}\|_2^2)\|\mathcal{W}^1-\mathcal{W}^2\|_F^2, 
	\end{aligned}
	\end{equation*}       
	where the last inequality is due to Cauchy-Schwarz inequality.
	Similarly, we can bound $\gamma^i$ as follows:
	\begin{equation*}
	\begin{aligned}
	\|\gamma^i \|_2^2 &\!\!\leq \sum_{j=1}^{n} \|\mathcal{K}_j \|_2^2 L_{\Psi}^2\! \Big(\!\sum_{k=1}^{d}(\mathcal{W}_k^1x_i\!-\!\mathcal{W}_k^2x_i)^2 \!+\! (\mathcal{W}_j^1x_i\!-\!\mathcal{W}_j^2x_i)^2\!\Big)\\
	& \!\!\leq \sum_{j=1}^{d} \|\mathcal{K}_j \|_2^2 L_{\Psi}^2 \Big(\|\mathcal{W}^1x_i-\mathcal{W}^2x_i\|_2^2 \!+\! d(\mathcal{W}_j^1x_i\!-\!\mathcal{W}_j^2x_i)^2\!\Big)  \\
	& \!\!\leq (dU_{\mathcal{K}}^2 + d^2\|\mathcal{K}_j^{\max}\|_2^2) L_{\Psi}^2\|\mathcal{W}^1x_i-\mathcal{W}^2x_i\|_2^2\\
	& \!\! \leq (dU_{\mathcal{K}}^2 + d^2\|\mathcal{K}_j^{\max}\|_2^2) L_{\Psi}^2 \|x_i\|_2^2 \|\mathcal{W}^1-\mathcal{W}^2\|_F^2 \\
	& \!\! \leq 2dU_{\mathcal{K}}^2 L_{\Psi}^2 \|x_i\|_2^2 \|\mathcal{W}^1-\mathcal{W}^2\|_F^2, \\
	\end{aligned}
	\end{equation*}
   where the second inequality is by  Assumption \ref{Assumptipon:2}.
	Combining the above results, we have 
	\begin{equation*}
	\begin{aligned}
	&\| \nabla_\mathcal{W} \mathcal{F}(\mathcal{W}^1, \mathcal{K}) - \nabla_\mathcal{W} \mathcal{F}(\mathcal{W}^2, \mathcal{K}) \|_F\\
	\leq & \frac{1}{N} \sum_{i=1}^{N}\|\mathcal{K}\|_F\|\alpha_i^{\mathcal{W}^1}-\alpha_i^{\mathcal{W}^2} \|_2 \|x_i\|_2\\
	\leq & \sqrt{2d}\|\mathcal{K}\|_FL_{\Psi} \frac{\sum_{i=1}^{N}\|x_i\|_2 \Delta_i}{N}\|\mathcal{W}^1-\mathcal{W}^2\|_F, \\
	\end{aligned}
	\end{equation*}
	where $\Delta_i = \sqrt{(1+dU_{\mathcal{K}}^{2})\|x_i\|^2_2 + \|x_{i+1}\|_2^2}$. 
\end{proof}

Similarly, $\mathcal{F}$ has Lipschitz-continuous gradient with respect to  the parameter $\mathcal{K}$ of the linear mapping. 
\begin{lemma} \label{lemma:2}
	Under Assumption \ref{Assumption:1} and assume that the basis function is bounded by $h$, we have 
	\begin{equation*}
	\|\nabla_{\mathcal{K}} \mathcal{F}(\mathcal{W}, \mathcal{K}^1) - \nabla_{\mathcal{K}} \mathcal{F}(\mathcal{W}, \mathcal{K}^2) \|_F \leq L_{\mathcal{K}} \|\mathcal{K}^1 - \mathcal{K}^2 \|_F
	\end{equation*}
	with  $L_{\mathcal{K}}=dh^2$.
\end{lemma}
\begin{proof}
	The gradient
	\begin{equation*}
	\begin{aligned}
	\nabla_{\mathcal{K}} \mathcal{F}(\mathcal{W},\mathcal{K}) = \frac{1}{N} \sum_{i=1}^{N}(\mathcal{K} \psi(\mathcal{W}x_i) - \psi(\mathcal{W}x_{i+1}))\psi(\mathcal{W}x_i)^T, 
	\end{aligned}
	\end{equation*}
	and 
	\begin{equation*}
	\begin{aligned}
	&~~~~\|\nabla_{\mathcal{K}} \mathcal{F}(\mathcal{W}, \mathcal{K}^1) - \nabla_{\mathcal{K}} \mathcal{F}(\mathcal{W}, \mathcal{K}^2) \|_F \\  
	&= \frac{1}{N}\|\sum_{i=1}^{N} (\mathcal{K}^1 - \mathcal{K}^2) \psi(\mathcal{W}x_i)\psi(\mathcal{W}x_i)^T\|_F \\
	&\leq \frac{1}{N} \|\mathcal{K}^1 - \mathcal{K}^2\|_F \|\sum_{i=1}^{N}\psi(\mathcal{W}x_i)\psi(\mathcal{W}x_i)^T\|_F\\
	&\leq  \frac{1}{N}\|\mathcal{K}^1 - \mathcal{K}^2\|_F \sqrt{d^2(Nh^2)^2}\\
	&= dh^2\|\mathcal{K}^1 - \mathcal{K}^2\|_F.
	\end{aligned}
	\end{equation*}
\end{proof}
With Lemmas \ref{lemma:1} and \ref{lemma:2}, we now show that  Algorithm \ref{algo:0} will converge to a critical point with convergence rate $O(\frac{1}{T})$ or $O(\frac{1}{\log T})$. 

\begin{theorem} \label{theorem:1}
	Under Assumptions \ref{Assumption:1} and \ref{Assumptipon:2}, Algorithm \ref{algo:0}  for the Koopman operator learning will converge to a critical point. With constant learning rate $\eta \leq \min(\frac{2}{L_{\mathcal{W}}},\frac{2}{L_{\mathcal{K}}})$, its convergence rate is $O(\frac{1}{T})$; and with diminishing learning rate $\eta_t = \frac{1}{t+1}$, its convergence rate is $O(\frac{1}{\log T})$. 
\end{theorem} 

\begin{proof}
	Since the objective function is Lipschitz gradient continuous with respect to $\mathcal{K}$, the descent lemma \cite{nesterov2013introductory} can be applied and we have   
	\begin{align}
		&\mathcal{F}(\mathcal{W}^{t},\mathcal{K}^{t+1}) \nonumber \\
		\leq& \mathcal{F}(\mathcal{W}^{t},\mathcal{K}^{t})+ \tr(\nabla_{\mathcal{K}} \mathcal{F}(\mathcal{W}^{t},\mathcal{K}^{t})^{T}(\mathcal{K}^{t+1} - \mathcal{K}^{t}))\nonumber \\
		&~~~~+ \frac{L_{\mathcal{K}}}{2}\|\mathcal{K}^{t+1} \!\!-\! \mathcal{K}^{t}\|_{F}^{2}\nonumber\\
		=& \mathcal{F}(\mathcal{W}^{t},\mathcal{K}^{t}) -\eta_{\mathcal{K}}\tr(\nabla_{\mathcal{K}} \mathcal{F}(\mathcal{W}^{t},\mathcal{K}^{t})^{T}\nabla_{\mathcal{K}} \mathcal{F}(\mathcal{W}^{t},\mathcal{K}^{t})) \nonumber\\
		&~~~~+ \frac{L_{\mathcal{K}}}{2}\|\mathcal{K}^{t+1} \!\!-\! \mathcal{K}^{t}\|_F^{2}\nonumber\\
		=& \mathcal{F}(\mathcal{W}^{t},\mathcal{K}^{t}) + \left(\frac{\eta_{\mathcal{K}}^{2}L_{\mathcal{K}}}{2}-\eta_{\mathcal{K}}\right)\|\nabla_{\mathcal{K}} \mathcal{F}(\mathcal{W}^{t},\mathcal{K}^{t})\|_F^{2}, \label{equ:5}
	\end{align}
	where $\tr$ denotes the trace of the matrix. The first equality is due to the gradient update of $\mathcal{K}^{t}$ and the second equality is by the fact that $\tr(A^TA) = \|A\|_{F}^{2}$.   
	
	 As for the basis function's parameter $\mathcal{W}$, we can have the similar result since the objective function is Lipschitz gradient continuous with respect to $\mathcal{W}$:
		 \begin{align}
			&\mathcal{F}(\mathcal{W}^{t+1},\mathcal{K}^{t+1}) \nonumber\\
			\leq & \mathcal{F}(\mathcal{W}^{t},\mathcal{K}^{t+1}\!) \!+\! \left(\!\frac{\eta_{\mathcal{W}}^{2}L_{\mathcal{W}}}{2}\!-\!\eta_{\mathcal{W}}\!\!\right)\|\nabla_{\mathcal{W}} \mathcal{F}(\mathcal{W}^{t}\!,\!\mathcal{K}^{t+1})\|_F^{2}.\label{equ:6}
		\end{align}
	So by equations \eqref{equ:5} and \eqref{equ:6}, we have the following for each complete update from $(\mathcal{W}^{t},\mathcal{K}^t) \rightarrow (\mathcal{W}^{t+1},\mathcal{K}^{t+1})$: 
	\begin{align}
		&\mathcal{F}(\mathcal{W}^{t+1},\mathcal{K}^{t+1}) \nonumber\\
		\leq & \mathcal{F}(\mathcal{W}^{t},\mathcal{K}^{t}) + 
		\left(\frac{\eta_{\mathcal{K}}^{2}L_{\mathcal{K}}}{2}-\eta_{\mathcal{K}}\right)\|\nabla_{\mathcal{K}} \mathcal{F}(\mathcal{W}^{t},\mathcal{K}^{t})\|_F^{2} \nonumber\\
		 &~~~+ \left(\frac{\eta_{\mathcal{W}}^{2}L_{\mathcal{W}}}{2}-\eta_{\mathcal{W}}\right)\|\nabla_{\mathcal{W}} \mathcal{F}(\mathcal{W}^{t},\mathcal{K}^{t+1})\|_F^{2}. \label{equ:7}
	\end{align} 
	We sum both sides of inequality \eqref{equ:7} from $t = 1,\cdots, T$ and obtain
	\begin{align}
		&\mathcal{F}(\mathcal{W}^{T+1},\mathcal{K}^{T+1}) \nonumber\\
		\leq & \mathcal{F}(\mathcal{W}^{0},\mathcal{K}^{0}) + 
	\sum_{t=0}^{T}	\left(\frac{\eta_{\mathcal{K}}^{2}L_{\mathcal{K}}}{2}-\eta_{\mathcal{K}}\right)\|\nabla_{\mathcal{K}} \mathcal{F}(\mathcal{W}^{t},\mathcal{K}^{t})\|_F^{2} \nonumber\\
		&~~~+ 	\sum_{t=0}^{T}\left(\frac{\eta_{\mathcal{W}}^{2}L_{\mathcal{W}}}{2}-\eta_{\mathcal{W}}\right)\|\nabla_{\mathcal{W}} \mathcal{F}(\mathcal{W}^{t},\mathcal{K}^{t+1})\|_F^{2}.	
	\end{align} 
    \textbf{(1) Constant learning rate}
    
    If we choose the constant stepsize, e.g., $\eta_{\mathcal{W}} = \eta_{\mathcal{K}} =  \eta$, $ 0 < \eta < \min(\frac{2}{L_{\mathcal{W}}}, \frac{2}{L_\mathcal{K}})$ and let $L = \max(\eta_{\mathcal{W}},\eta_{\mathcal{K}})$, $S = \eta - \frac{L\eta^2}{2}$, and we can bound the gradients as follows: 
    \begin{align}
    	&\sum_{t=0}^{T}\|\nabla_{\mathcal{K}} \mathcal{F}(\mathcal{W}^{t},\mathcal{K}^{t})\|_F^{2} + \|\nabla_{\mathcal{W}} \mathcal{F}(\mathcal{W}^{t},\mathcal{K}^{t+1})\|_F^{2} \nonumber\\ 
    	\leq & \frac{(\mathcal{F}(\mathcal{W}^{0},\mathcal{K}^{0}) - \mathcal{F}(\mathcal{W}^{T+1},\mathcal{K}^{T+1}))}{S}. \label{equ:9}
    \end{align}
    One can see that, each term on the right is non-negative and their summation is bounded by some constant. 
    We can conclude that the alternating optimization algorithm will converge asymptotically to one critical point even \textbf{without} optimal tracker $\mathcal{K}^*$ and $\mathcal{W}^*$ when $T \rightarrow \infty.$

    Based on inequality \eqref{equ:9}, one can bound the minimum gradients up to $T$ for Algorithm \ref{algo:0} as follows: 
    \begin{align}
    	&\min_{t= 0,\cdots, T} \|\nabla_{\mathcal{K}} \mathcal{F}(\mathcal{W}^{t},\mathcal{K}^{t})\|_F^{2} + \|\nabla_{\mathcal{W}} \mathcal{F}(\mathcal{W}^{t},\mathcal{K}^{t+1})\|_F^{2} \nonumber\\
    	\leq& \frac{(\mathcal{F}(\mathcal{W}^{0},\mathcal{K}^{0}) - \mathcal{F}(\mathcal{W}^{T+1},\mathcal{K}^{T+1}))}{ST} 	\leq\frac{2R}{ST}.\nonumber 
    \end{align}
	\textbf{(2) Diminishing learning rate}
	
	If we choose the diminishing learning rate, e.g., $\eta^{t} = \frac{1}{t+1}$, the result becomes
	\begin{align}
			&\min_{t= 0,\cdots, T} \|\nabla_{\mathcal{K}} \mathcal{F}(\mathcal{W}^{t},\mathcal{K}^{t})\|_F^{2} + \|\nabla_{\mathcal{W}} \mathcal{F}(\mathcal{W}^{t},\mathcal{K}^{t+1})\|_F^{2} \nonumber\\	
			\leq& \frac{(\mathcal{F}(\mathcal{W}^{0},\mathcal{K}^{0}) - \mathcal{F}(\mathcal{W}^{T+1},\mathcal{K}^{T+1}))}{\sum_{t=0}^{T}\left( \eta^t - \frac{L(\eta^t)^2}{2}     \right)} \nonumber 	.	\end{align}
	We know
	\begin{align}
		&\sum_{t=0}^{T}\left( \eta^t - \frac{L(\eta^t)^2}{2}     \right) = \sum_{t=0}^{T} \left( \frac{1}{t+1} - \frac{L}{2(t+1)^2}   \right) \nonumber\\
		\geq & \ln(T+2) \!-\! \frac{L}{2} \!-\! \sum_{t=1}^{T} \frac{L}{2t(t+1)} = \ln( T\!+\!2) \!-\! L + \frac{L}{2(T+1)}.
	\end{align}
	So for diminishing stepsize, we can obtain 
	\begin{align}
		&\min_{t= 0,\cdots, T} \|\nabla_{\mathcal{K}} \mathcal{F}(\mathcal{W}^{t},\mathcal{K}^{t})\|_F^{2} + \|\nabla_{\mathcal{W}} \mathcal{F}(\mathcal{W}^{t},\mathcal{K}^{t+1})\|_F^{2} \nonumber\\	
		\leq & \frac{2R}{O(\ln T)} \nonumber.
	\end{align}
	
	We see that the constant stepsize has the better convergence rate than the diminishing stepsize with the help of optimal tracker. Both cases show that 
	\begin{align*}
	    	\|\nabla_{\mathcal{K}} \mathcal{F}(\mathcal{W}^{*},\mathcal{K}^{*})\|_F^{2} + \|\nabla_{\mathcal{W}} \mathcal{F}(\mathcal{W}^{*},\mathcal{K}^{*})\|_F^{2} \rightarrow 0.
	\end{align*}
\end{proof}
\section{Distributed Koopman Learning}\label{sect:distributed}
We now develop an algorithm to handle the learning problem for the Koopman operator of high dimensional nonlinear dynamical systems.  Even if there are only a thousand states in the underlying nonlinear system,  the dimension of the dictionary functions explodes exponentially with the number of states.  Memory constraints thus make it infeasible to train a Koopman operator using a centralized or stand-alone computing node.  This motivates the derivation of a scalable, distributed approximation algorithm to relieve this problem.
\begin{assumption} \label{Assup:3}
	The basis function $\psi(x)$ can be decomposed or approximated by $[\psi_1(x^1),\cdots, \psi_q(x_q)]^{\top},$ where $\psi_i: \mathbb{R}^{d_i} \rightarrow \mathbb{R}^{m_i}$ is the new basis function for $x^i$ and $x_i$ is a subset of $x$ with $x = [x^1,\cdots,x^q]^{\top}$. 
\end{assumption}

Based on Assumption \ref{Assup:3}, we can reformulate the centralized Koopman objective function as 

\begin{align}
&\mathcal{F}(\mathcal{W},\mathcal{K}) \nonumber\\
=& \frac{1}{2N}\sum_{i=1}^{N}\left\|\begin{bmatrix}	 
\psi_1(\mathcal{W}_1 x_{i+1}^1) \\
\vdots\\
\psi_q(\mathcal{W}_q x_{i+1}^q)
\end{bmatrix}
\!\!-\!\!
\begin{bmatrix}
\mathcal{K}_{11} \cdots \mathcal{K}_{1q} \\
\vdots ~~\ddots ~~\vdots \\
\mathcal{K}_{q1} \cdots \mathcal{K}_{qq} 
\end{bmatrix}\!\! \begin{bmatrix}
\psi_1(\mathcal{W}_1 x_i^{1}) \\
\vdots\\
\psi_q(\mathcal{W}_q x_i^q)
\end{bmatrix}\right\|_{F}^2 \nonumber
\end{align}

Our distributed Koopman learning's structure is as follows. Denote by $\mathcal{Q} = [1,\cdots,q]$ the set of computation nodes which can communicate with each other. For each computation node $i \in \mathcal{Q}$, it only store part of the data set $\{(x^i_j,x^i_{j+1})|j = 1,\cdots N\}$, its corresponding row and column of Koopman operator $\{\mathcal{K}_{ij},\mathcal{K}_{ji}|j \in \mathcal{Q}\}$ and its basis function $\psi_i$. 

For node $i$, its gradient will compose two parts. The first part can be calculated based on  its own knowledge. Another part needs the information from other nodes. We first define by $e^i \in \mathbb{R}^{m_i}$ the error term for node $i$ with data point $(x_j^{i},x_{j+1}^{i})$, where 
\begin{align}
e_j^{i} = \psi_i(\mathcal{W}_i x_{j+1}^{i}) - \begin{bmatrix}
\mathcal{K}_{i1},\cdots, \mathcal{K}_{iq}
\end{bmatrix}\begin{bmatrix}
\psi_1(\mathcal{W}_1 x_j^{1})\\
\vdots\\
\psi_q(W_q x_j^{q})
\end{bmatrix}, \nonumber
\end{align}
and define by $J(\psi_i(\cdot)) \in \mathbb{R}^{m_i\times d_i}$ the Jacobi matrix of function $\psi_i$, we then have the following distributed Koopman learning algorithm shown in Algorithm \ref{algo:1}.

\begin{algorithm}
	\caption{Distributed Koopman Operator Learning}\label{algo:1}
		\textbf{Initialization:}\\
	\text{~~~node $i$ randomly initilizes its $W_i$, $\mathcal{K}_{ik}~\forall i,k \in \mathcal{Q}$.}\\	
	\While{Not Converge}{ 
		$A_i = 0, B_i = 0, C_i = 0, \forall i \in \mathcal{Q}.$\\
		\For{$j=1;j \leq N;j=j+1$}
		{
			node $i$ calculates $\psi_i(\mathcal{W}_{i} x_j^{i}),\!\psi_i(\mathcal{W}_{i} x_{j+1}^{i})$,$\forall i \in \mathcal{Q}$.	\\
			node $i$ broadcasts $S^i_j = \psi_i(\mathcal{W}_{i} x_j^{i})$,$\forall i \in \mathcal{Q}$.\\
			node $i$ calculates $e_j^i$, sends $S_{iv}^{'} = \mathcal{K}_{iv}^Te_j^i$ to node $v$, $\forall i,v \in \mathcal{Q}$.\\
			$A_i = A_i + J(\psi_i(\mathcal{W}_i x_{j+1}^{i}))^{\top}e_j^{i}, \forall i \in \mathcal{Q}$\\
			$B_i = B_i - J(\psi_i(\mathcal{W}_i x_j^i))^{\top}\sum_{k \in \mathcal{Q}}S_{ki}^{'},\forall i \in \mathcal{Q}$.\\
			$C_i = C_i + e_j^i \vect $\footnotemark$ [S_j^1,\cdots,S_j^q]$
		}
		\For{$i \in \mathcal{Q}$ }{
			$\mathcal{W}_{i} \leftarrow \mathcal{W}_{i} - \eta_{\mathcal{W}_i}\frac{1}{N}(A_i+B_i)$.\\
			$[\mathcal{K}_{i1}, \cdots, \mathcal{K}_{iq}] \leftarrow [\mathcal{K}_{i1}, \cdots, \mathcal{K}_{iq}] - \frac{\eta_{\mathcal{K}_i}}{N}C_i$  
		}
	}	
\end{algorithm}
\footnotetext{$\vect$ means vectorization of matrix which converts the matrix into a column vector. }
For our distributed Koopman operator learning Algorithm $\ref{algo:1}$, line 6-8 is the communication stage, each computation node $i$ calculates its result in the lifted dimensional space and then broadcast it over a communication network. After the communication, the information is enough to compute local error term $e_{j}^{i}$, and node $i$ send $S_{iv}^{'}$ to node $v$ (line 8). Here the communication stage ends and computation stage (line 9-11) begins. $A_i,B_i,C_i$ will sum up all the information for each data point. The last is update stage with gradient descent method (line 14-15). Based on this distributed algorithm and Assumption \ref{Assup:3}, we can prove this is equivalent to the centralized gradient descent algorithm. 
\begin{lemma} \label{theorem:2}
Under Assumption \ref{Assup:3}, the distributed Koopman learning in  Algorithm \ref{algo:1} is equivalent to the following update:
	\begin{equation}
	\begin{aligned} \label{equ:11}
		\mathcal{K}^{t+1} &= \mathcal{K}^{t} - \eta_{\mathcal{K}} \nabla_{\mathcal{K}} \mathcal{F}(\mathcal{W}^t, \mathcal{K}^t), \\
		\mathcal{W}^{t+1} &= \mathcal{W}^{t} - \eta_{\mathcal{W}} \nabla_{\mathcal{W}} \mathcal{F}(\mathcal{W}^t, \mathcal{K}^{k}).
	\end{aligned}
	\end{equation}  
\end{lemma}
\begin{proof}
	Based on line 9-11 in Algorithm \ref{algo:1}, one can verify the following after updating all the data points: 
	\begin{align*}
		A_i &= \sum_{j=1}^{N} J(\psi_i(\mathcal{W}_ix_{j+1}^{i}))^{\top}e_j^{i},\\
		B_i &= -\sum_{j}^{N}J(\psi_i(\mathcal{W}_i x^i_{j}))^{\top}\sum_{k \in \mathcal{Q}}\mathcal{K}_{iv}^Te_j^i,\\
		C_i & = \sum_{j}^{N}e_j^i \vect[S_j^1,\cdots,S_j^q].
	\end{align*}
Compared to gradients of $\mathcal{F}(\mathcal{W},\mathcal{K})$, we can find that 
\begin{align*}
	\frac{1}{N}(A_i + B_i) &= \nabla_{\mathcal{W}} \mathcal{F}(\mathcal{W}^t, \mathcal{K}^{k}),\\
	\frac{1}{N}C_i & = \nabla_{\mathcal{K}} \mathcal{F}(\mathcal{W}^t, \mathcal{K}^t).
\end{align*}
So the update stage (line 14-15) is the same with equation \eqref{equ:11}, which finishes the proof.
\end{proof}
\begin{remark}
	Our alternating Koopman operator learning (Algorithm \ref{algo:0}) can be regarded as nonlinear Gauss-Seidel iterations \cite{vrahatis2003linear}, while our distributed Koopman operator learning lies in the model of nonlinear Jacobi iteration \cite{vrahatis2003linear}. Here we choose nonlinear Jacobi iteration for distributed Koopman operator learning due to that nonlinear Jacobi iteration is \textbf{(1)} suitable for parallel computation and \textbf{(2)} with less communication overhead.     
\end{remark}

By lemma \ref{theorem:2}, we have the following convergence result for our distributed Koopman operator learning. 
\begin{theorem}
	Under Assumption \ref{Assumption:1}, \ref{Assumptipon:2} and \ref{Assup:3}, the distributed Koopman operator learning based on Algorithm \ref{algo:1} will converge to a critical point asymptotically.
\end{theorem}
\begin{proof}
	The equation \eqref{equ:11} is the case without optimal tracker of Algorithm \ref{algo:0}. The proof of theorem \ref{theorem:1} can be applied directly here (equation \eqref{equ:9}).
\end{proof}

	The advantages of the distributed Koopman learning over the centralized one are not only the scalability, e.g., the ability to handle the high dimensional nonlinear dynamics, but also the feasibility to adjust to different complexity of the partial observations. For example, if one partial observation $(x_j^i,x_{j+1}^i)$ is with complexity dynamic, we can increase the number of basis function.

On the other hand, Algorithm \ref{algo:1} for distributed Koopman learning is under the ideal synchronous model. Although the computation of each node $i$ is parallel, the computation will not start until the broadcast process finishes. This can lead to significant inefficiency in the distributed Koopman learning when, e.g., one node has very limited  communication capability so that all other nodes need wait for this node. Also, one packet loss will lead to all nodes waiting for the resending. However, it is easily to extend Algorithm \ref{algo:1} to handle asynchronous mode as shown in Algorithm \ref{algo:2}. Each node will store the received information $(S_j^{i}, S_{iv}^{'})$ in its memory, updating it when new information arrives. Once the computation node comes to computation stage, it will directly use the information stored in the memory instead of waiting for the newest information.
\begin{algorithm}
	\caption{Asynchronous Distributed Koopman Learning}\label{algo:2}
	
	lines 1-7 of Algorithm \ref{algo:1}.\\
			node $i$ calculates $e_j^i$ based on the current $S_l$ in the memory, sends $S_{iv}^{'} = \mathcal{K}_{iv}^Te_j^i$ to node $v$, $\forall i,v,l \in \mathcal{Q}$.\\
			$A_i \leftarrow A_i + J(\psi_i(\mathcal{W}_i x_{j+1}^{i})^{\top}e_i, \forall i \in \mathcal{Q}$\\
			$B_i \leftarrow B_i - J(\psi_i(\mathcal{W}_i x_j^i);)^{\top}\sum_{k \in \mathcal{Q}}S_{ki}^{'}$  based on the current $S_{ki}$ in the memory, $\forall i \in \mathcal{Q}$.\\
			$C_i \leftarrow C_i + e_j^i \vect[S_1,\cdots,S_q]$ based on the current $S_l$ in the memory, $\forall l \in \mathcal{Q}$.\\
	lines 13-16 of Algorithm \ref{algo:1}.
\end{algorithm}

For the asynchronous version of gradient descent algorithm, existing work such as \cite{low1999optimization} (Theorem 2), \cite{bertsekas1989parallel} (Proposition 2.1), and \cite{liu2017proportional} (Theorem 7)  shows that the synchronous and asynchronous algorithms will converge to the same point once as long as  communication delay is bounded. Their proof applies to our problem with slight modification. 
\begin{lemma}{(\cite{low1999optimization},\cite{bertsekas1989parallel},\cite{liu2017proportional})}
	If the communication delay is bounded by some constant, with small enough stepsize, the asynchronous algorithm \ref{algo:2} will asymptotically converge to the same point as the synchronous one.
\end{lemma}

\section{Experiments}\label{sect:experiment}
We now evaluate the performance of the proposed alternating optimization and distributed algorithms. For each experiment, we sample some points from the real trajectory to prepare the training set and prediction set. Note that our prediction phase is multi-step prediction, i.e., given one initial state, our algorithm will predict the following trajectory using the $\mathcal{K}$-invariant property: $\psi(x_n) = \mathcal{K}^n \psi(x_0).$ 

To evaluate the performance of alternating optimization, we consider Van der Pol oscillator shown in Example \ref{exam:1}. 
\begin{example}{Van der Pol oscillator} \label{exam:1}
    \begin{align}
    \dot{x}_{1} &= \mu \left( x_1 - \frac{1}{3}x_1^3 - x_2\right) \\
    \dot{x_2}  &= \frac{1}{\mu}x_1
\end{align}
\end{example}

In this example, we choose $\mu = 0.5$. The number of date points we sampled is 600 with 400 for training and 200 for prediction. We construct a very simple network with one layer and 3 dimensions to learn the pattern of Van der Pol oscillator. The total training time is around 1.08s (i7-8700 CPU @ 3.20 GHz, 8 GB RAM) with 500 iterations and constant stepsize is 0.23. Fig. \ref{fig:1} shows the multi-step prediction result with alternating optimization method. One step prediction error is around $0.16\%$ and 200 step prediction error is around $1.89\%$.
\begin{figure}[ht!]
	\includegraphics[scale = 0.35]{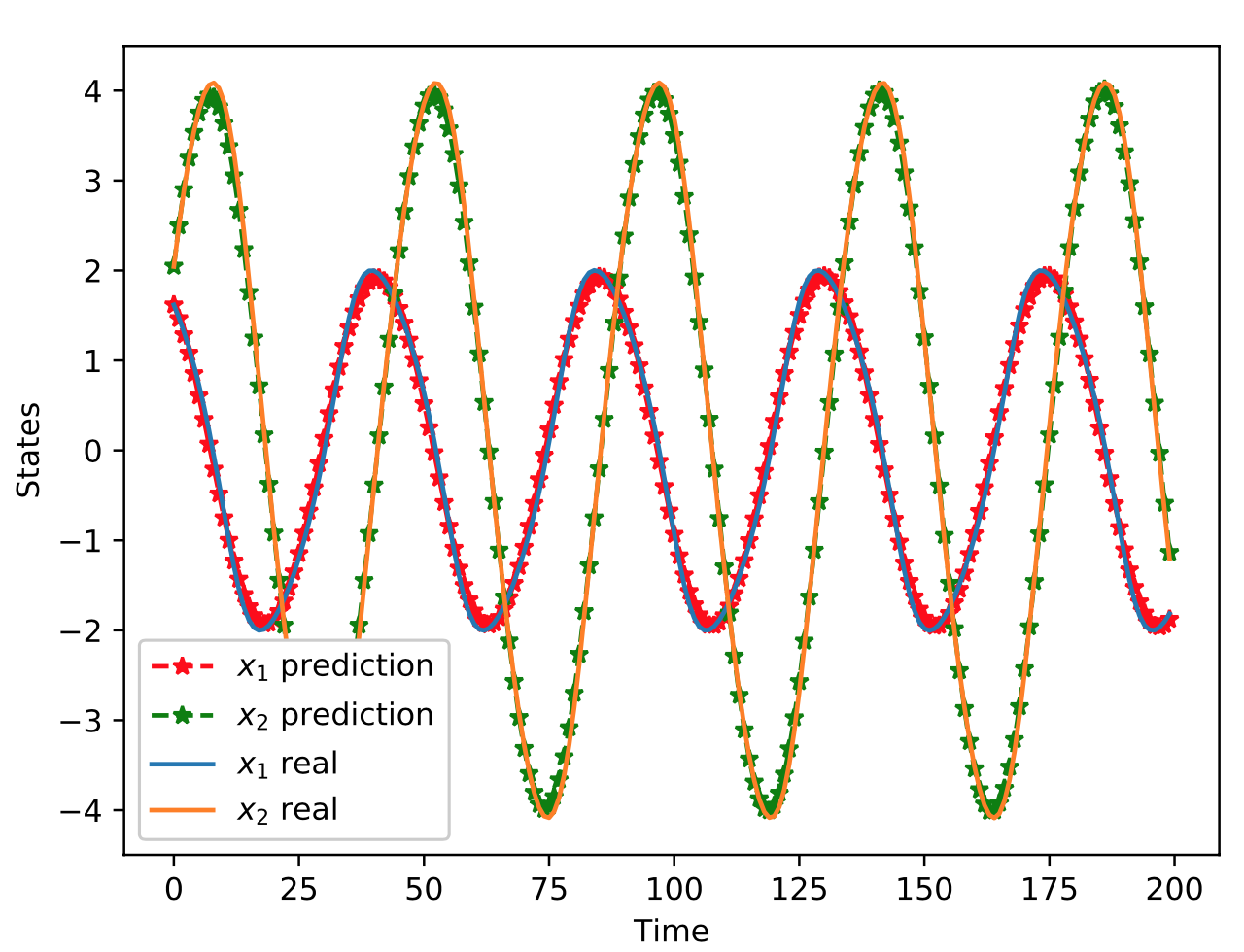}
	\caption{Alternating optimization for centralized Koopman operator learning with Van der Pol oscillator. In this experiments, only the points at time 0 are given. All the data points [1-200] are our predictions with Koopman learning.}
	\label{fig:1}
\end{figure}

\begin{example}{Glycolytic pathway} \label{exam:2}
\begin{align}
    \dot{x_1} &= J - \frac{k_1 x_1 x_6}{1+\left(\frac{x_6}{k_1}\right)^{q}} \\
    \dot{x_2} &= \frac{2k_1 x_1 x_6}{1+\left(\frac{x_6}{k_1}\right)^{q}} - k_2x_2(n-x_5) - k_6 x_2 x_5 \\
    \dot{x_3} &= k_2 x_2(n-x_5) - k_3 x_3(a-x_6)\\
    \dot{x_4} &= k_3 x_3(a-x_6) - k_4 x_4 x_5 - \kappa(x_4 - x_7)\\
    \dot{x_5} &= k_2 x_2(n-x_5) - k_4 x_4 x_5 - k_6 x_2 x_5\\
    \dot{x_6} &= -\frac{2k_1 x_1 x_6}{1+\left(\frac{x_6}{k_1}\right)^{q}} + 2k_3 x_3(a-x_6) - k_5 x_6\\
    \dot{x_7} &= \phi \kappa (x_4 - x_7) - k x_7
\end{align}
\end{example}

\begin{figure*}
    \centering
	\includegraphics[width=0.8\textwidth]{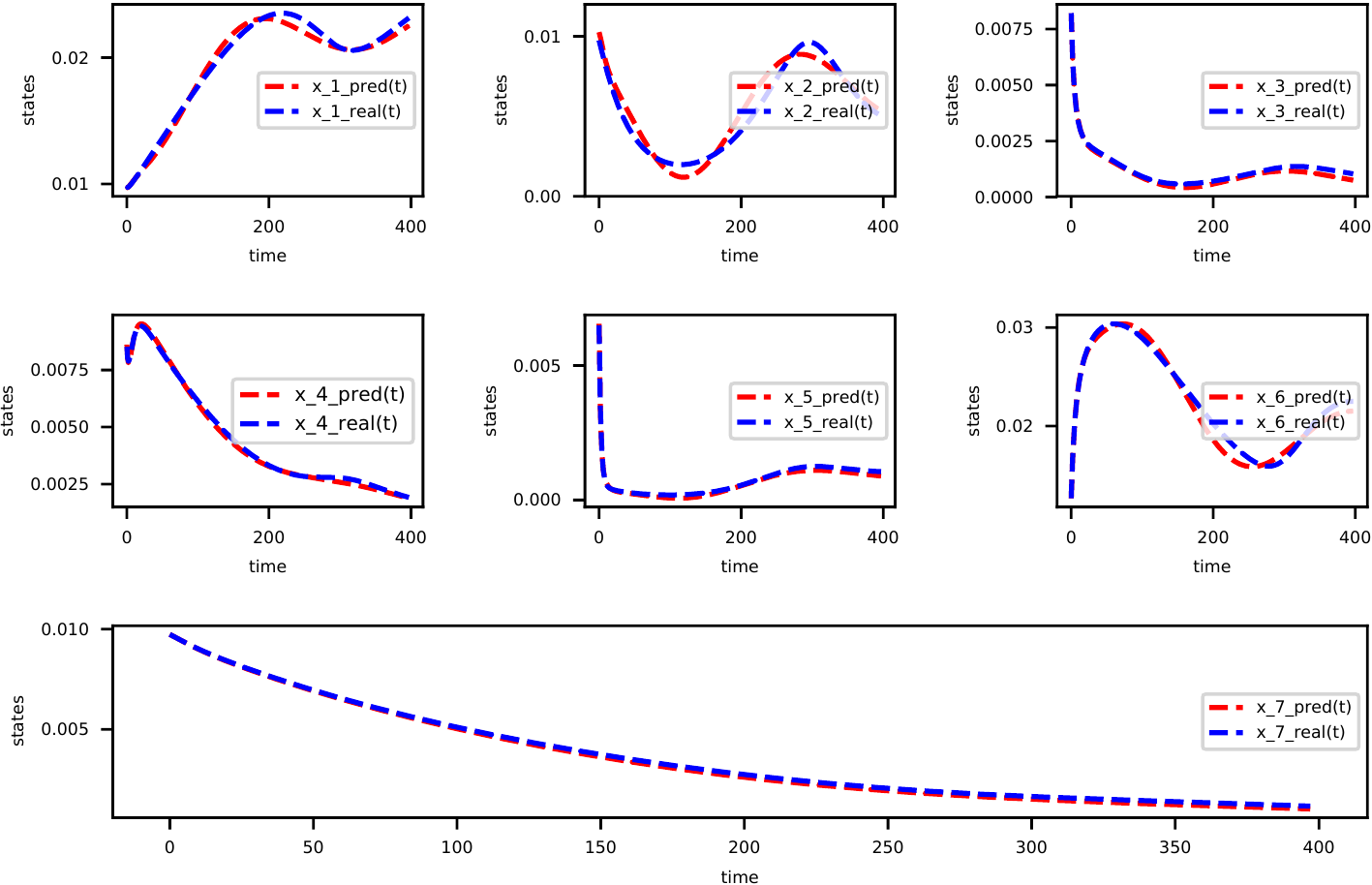}
	\caption{Distributed Koopman learning for Glycolytic pathway.}
	\label{fig:2}
\end{figure*} 

Our distributed Koopman operator learning is implemented on a larger nonlinear dynamical system shown in Example \ref{exam:2}, namely the glycolysis network from cellular biology \cite{daniels2015efficient}. We adopt the parameter setting: $J = 2.5, a= 4, n = 1, k_1 = 0.52, \kappa = 13, \phi = 0.1, q =4, k = 1.8, k_1 = 100, k_2 = 6, k_3 = 16, k_4 = 100, k_5 = 1.28, k_6 = 12$ from \cite{daniels2015efficient}. 1000 data points are sampled from the real trajectory with 600 points for training and 400 for prediction. We create 7 threads to simulate the distributed learning and each thread only learn the dynamic pattern of one state by a simple 3-layer neural network with 15 dimensions. The total training time is 400.4s with 10000 iterations. Results for each state is shown in Fig. \ref{fig:2}. One step error is around 0.02\% and 400 step prediction error is around 2.7\%.

We see that our alternating optimization and distributed algorithms both achieve good performance with multi-step prediction.  Even though partial state measurements are provided for training, the trained distributed Koopman operator is able to predict the behavior of the glycolysis network over 400 steps.  Further, these results provide a glimmer of hope for whole cell network modeling, using strategically placed reporter libraries that provide partial measurements of an entire transcriptome \cite{hasnain2019optimal,ward2009}.

\section{Conclusion}  \label{sect:conclusion}
We have proposed an alternating optimization algorithm to the nonconvex Koopman operator learning problem for nonlinear dynamic systems. We prove that the proposed algorithm converges to a critical point with rate $O(1/T)$ and $O(\frac{1}{\log T})$ for the constant and diminishing learning rates, respectively, under some mild conditions. To cope with the high dimensional nonlinear dynamical systems, we have further proposed a distributed Koopman operator learning algorithm with an appropriate communication mechanism. We show that the distributed Koopman operator learning is of the same convergence property with the centralized one if the basis functions are decomposable. Numerical experiments are provided to complement our theoretical results.

\section{Acknowledgments}
We thank Igor Mezic, Nathan Kutz, Robert
Egbert, Bassam Bamieh, Sai Nandanoori Pushpak, Sean Warnick, Jongmin Kim, Umesh Vaidya, and Erik Bollt for stimulating conversations. Any opinions,
findings and conclusions or recommendations expressed in
this material are those of the author(s) and do not necessarily
reflect the views of the Defense Advanced Research Projects
Agency (DARPA), the Department of Defense, or the United
States Government. This work was supported partially by
a Defense Advanced Research Projects Agency (DARPA)
Grant No. DEAC0576RL01830 and an Institute of Collaborative Biotechnologies Grant.

\bibliographystyle{unsrt}
\bibliography{reference}

\end{document}